\documentclass[10pt,a4paper]{article}
\usepackage{fullpage}
\usepackage{amsfonts,amsmath,amssymb}
\usepackage{amsthm}
\usepackage{graphicx}\usepackage{sectsty}
\usepackage{authblk}

\theoremstyle{plain}
\newtheorem{theorem}{Theorem}[section]
\newtheorem{lemma}[theorem]{Lemma}
\newtheorem{proposition}[theorem]{Proposition}

\theoremstyle{definition}

\theoremstyle{remark}

\numberwithin{equation}{section}

\sectionfont{\large}
\newenvironment{acknowledgement}[1][Acknowledgement
]{\begin{trivlist} \item[\hskip \labelsep {\bfseries
#1}]}{\end{trivlist}}

\begin{document}
\title{Recovering the Topology in One Point Interaction Problem on Extended Non-Local Star Graphs}
\author[1]{Lung-Hui Chen}
\affil[1]{\footnotesize General Education Center, Ming Chi University of Technology, New Taipei City, 24301, Taiwan.}
\maketitle
\begin{abstract}
The author studies the inverse spectral problem of Sturm-Liouville operator on a star-like graph. To this star-like graph centered at the origin as its vertex, there are attached $m$ edges that imposed the Sturm-Liouville operator with certain non-local potential functions with some suitable local boundary value conditions. At the vertex, we consider one point interaction condition at vertex to model a network that fixed on the end of the edges on the graph. The vibration and flow changes are monitored at that vertex which serves as certain control/regulation center. The author shows that the system is solvable under very necessary conditions. It is crucial to recover  the topology of the network. In this paper, author constructs the special solution edge by edge and point to point.
\end{abstract}
\section{Introduction}
This work explores the inverse spectral analysis of Sturm-Liouville operators situated on star-shaped metric graphs. On such graphs, where multiple edges emanate from a central vertex, provide an effective mathematical model for analyzing signal transmission and dynamic processes within structured networks. The application of Sturm-Liouville theory in this context has become increasingly relevant in bridging mathematical physics with real-world engineering systems such as telecommunications, energy grids, and transport logistics, as noted in  studies \cite{Freiling,Kostrykin,Kuchment,Kurasov,Novikov,Pokornyi}.

The behavior of wave propagation and diffusion on these networks can be described by differential operators acting along each edge, governed by local boundary conditions and damped by non-local potential terms. At the central vertex, a fixed-point (frozen argument) condition is imposed, simulating a centralized monitoring or control mechanism. This setting allows us to model systems where information or energy must flow efficiently and reversibly across the distributed structure.

The non-local boundary conditions are not merely technical features, but also contribute to the stability and physical realization of certain engineering modeling. In particular, they help explain the physical principles such as the conservation of energy which in turn the mathematical uniqueness of solutions of the state. The investigation presented here demonstrates that, under appropriate local or non-local conditions, the inverse spectral problem admits an inverse unique and solvable formulation.

\par
The most interesting aspect of such modeling is its connection to the theory of quantum graphs, which studies the physical dynamics on metric graphs governed by various sorts of differential operators \cite{Alb,Bon,Freiling,Gerasimenko,Kottos,Naimark}, along with many different kinds of boundary conditions which often compromise the energy conservation laws or even certain commercial requirements in real world scenarios. The boundary condition acts as an integral part of the modeling these specific problems. It plays a role in many engineering projects. We are interested to formulate and to solve a few number of direct and inverse spectral problems and scattering problems on quantum graphs \cite{Belishev,Carlson,Kostrykin,Kurasov,Kurasov2,Pivovarchik1,Yurko1,Yurko2}.
Specifically, in this paper, we study the inverse spectral problems for a Schr\"{o}dinger operator with non-local potential functions on a star graph. Previously, the inverse spectral problems on a finite graph for the Sturm-Liouville problems with non-local potentials or with various boundary value conditions were considered in \cite{Alb,Alb2,Bon,Nizhnik1,Nizhnik2,Pivovarchik1,Yang}.
\par
Let us consider the following star graph $T$: The central vertex of $T$ is located at the origin $0\in\mathbb{R}^{2}$ and connected by 
$m$ edges, each sharing the origin as their common vertex, which we denoted as $V_{0}$.  Each edge $e_{j}$ has length $l_{j}$, $j=1,2,\ldots,m$ with end points denoted as $v_{j}$, $j=1,2,\ldots,m$.
    \par
On each edge $e_{j}$, we assume that there exists the function $\psi_{j}(x)\in W_{2}^{2}(0,l_{j})$ such that the function $\psi_{j}(x)$ satisfies the following eigenvalue problem with complex-valued non-local potentials $q_{j}(x)\in L^{2}(0,l_{j})$:
\begin{eqnarray}\label{1.1}
-\frac{d^{2}}{dx^{2}}\psi_{j}(x)+q_{j}(x)\psi_{j}(0)=\lambda\psi_{j}(x),\,0<x<l_{j},\,j=1,2,\ldots,m,
\end{eqnarray}
in which $m\geq2$, $\lambda\in\mathbb{C}$ is the spectral parameter,
and it is imposed with the boundary conditions
\begin{equation}\label{1.2}
\psi_{j}(l_{j})=0,\,j=1,2,\ldots,m;\,\psi_{1}(0)=\psi_{2}(0)=\cdots=\psi_{m}(0),
\end{equation}
and the non-local boundary value condition
\begin{equation}\label{1.3}
\sum_{j=1}^{m}\big\{\frac{d}{dx}\psi_{j}(0)-\int_{0}^{l_{j}}\psi_{j}(x)\overline{q_{j}(x)}dx\big\}=0.
\end{equation}
Also, $\lambda=(\frac{z\pi}{l_{j}})^{2}$, $z\in\mathbb{C}$, and the solution $\psi_{j}(x,z)$ depends on $(x,z)\in\mathbb{R}\times\mathbb{C}$ for each $j$.
We note that the inverse spectral problems is posed with a frozen argument at the vertex that is called a non-local point interaction in quantum mechanics, and  the equation actually reads as
\begin{equation}\label{1.4}
\sum_{j=1}^{m}\big\{\frac{d}{dx}\psi_{j}(0,z)-\int_{0}^{l_{j}}\psi_{j}(x,z)\overline{q_{j}(x)}dx\big\}=0,\,z\in\mathbb{C}.
\end{equation}
In this paper, we use the spectral data to recover the information on angular information between the edges $$\{e_{1},e_{2},\cdots,e_{m}\}.$$
To this end, we consider the extended closed graph from the star graph $T$, which is denoted as $\overline{T}$ that is obtained by connecting the end point $v_{j}$, $j=1,2,\ldots,m$, and $v_{m+1}=v_{1}$. Let us denote the angle between $e_{j}$ and $e_{j+1}$ as $\theta_{j}$, and the angle between $e_{m}$ and $e_{1}$ as $\theta_{m}$. Moreover, we call the new edges between the vertices $v_{j}$ and $v_{j+1}$ as $\bar{e}_{j}$ with length $\bar{l}_{j}$, and the edge between $v_{m}$ and $v_{1}$ as $\bar{e}_{m}$ with length $\bar{l}_{m}$. 
We set $\cup_{j=1}^{m}\bar{e}_{j}$ to be a cyclic metric graph by taking $v_{m+1}=v_{1}$.
Surely, $v_{j},\,j\geq1,$ is a vertex with edge $e_{j}$, $\bar{e}_{j-1}$, and $\bar{e}_{j+1}$. The extension graph $\overline{T}$ is trivially unique.
\par
Therefore, we modify the differential system~(\ref{1.1}),~(\ref{1.2}), and~(\ref{1.4}) on the presence of  $\bar{e}_{j'}$, $j'=1,2,\cdots,m.$ For $z\in\mathbb{C}$,
\begin{eqnarray}\label{1.5}
\left\{
  \begin{array}{llllll}\vspace{7pt}
-\frac{d^{2}}{dx^{2}}\psi_{j}(x,z)+q_{j}(x)\psi_{j}(0,z)=(\frac{z\pi}{l_{j}})^{2}\psi_{j}(x,z),\,x\in e_{j},\,j=1,2,\ldots,m;\\\vspace{7pt}
-\frac{d^{2}}{dx^{2}}\psi_{j'}(x,z)=(\frac{z\pi}{l_{j}})^{2}\psi_{j'}(x,z),\,x\in \bar{e}_{j'},\,j'=1,2,\ldots,m;\\\vspace{7pt}
\psi_{j}(l_{j},z)=0,\,j=1,2,\ldots,m;\\\vspace{7pt}
\psi_{1}(0,z)=\psi_{2}(0,z)=\cdots=\psi_{m}(0,z);\\\vspace{7pt}
\psi_{j'}(l_{j'}^{-},z)=\psi_{j'-1}(\bar{l}_{j'-1}^{-},z)=\psi_{j'}(0^{+},z)=0,\,j=1,\cdots,m;\\\vspace{7pt}
\sum_{j=1}^{m}\sum_{j'=1}^{m}\big\{\frac{d}{dx}\psi_{j}(0^{+},z)+\frac{d}{dx}\psi_{j'}(l_{j'}^{-},z)+\frac{d}{dx}\psi_{j'-1}(\bar{l}_{j'-1}^{-},z)+\frac{d}{dx}\psi_{j'}(0^{+},z)\\\hspace{80pt}-\int_{0}^{l_{j}}\psi_{j}(x,z)\overline{q_{j}(x)}dx\big\}=0,
  \end{array}
\right.
\end{eqnarray}
in which $\frac{d}{dx}\psi_{j'}(l_{j'}^{-},z)+\frac{d}{dx}\psi_{j'-1}(\bar{l}_{j'-1}^{-},z)+\frac{d}{dx}\psi_{j'}(0^{+},z)$ prepares for some kind Kirchhoff's condition at vertex $v_{j'}$ on $\cup_{j=1}^{m}\bar{e}_{j}$ and the last line is the non-local Kirchhoff's condition in the system which collects all the eigenvalues of~(\ref{1.5}). When using subscript $j'$, we mean to count the objects on $\cup_{j=1}^{m}\bar{e}_{j}$.
In this paper, we aim to recover the information on $\theta_{j}$, $j=1,2,\ldots,m$, while the knowledge of $l_{j}$ and $q_{j}(x)$, $j=1,2,\ldots,m$, is given.

\section{Special Solution and Characteristic Function}
In this section, we construct the special solutions for system~(\ref{1.5}).
We start with the Fourier-sine series in $L^{2}(0,l_{j})$:
$$q_{j}(x)=\sum_{n=1}^{\infty}q_{j,n}\sin[\frac{n\pi}{l_{j}} (l_{j}-x)]$$ with Fourier-Sine transform and its coefficients
\begin{eqnarray}\label{2.1}
&&\mathcal{F}(q_{j})(z)=\frac{2}{l_{j}}\int_{0}^{l_{j}}q_{j}(x)\sin[\frac{z\pi}{l_{j}} (l_{j}-x)]dx;\\
&&q_{j,n}=\frac{2}{l_{j}}\int_{0}^{l_{j}}q_{j}(x)\sin[\frac{n\pi}{l_{j}} (l_{j}-x)]dx.
\end{eqnarray}
Firstly, we consider a special solution of system~(\ref{1.1}) with $\lambda=(\frac{z\pi}{l_{j}})^{2}$ to meet all conditions in~(\ref{1.1}). We choose
\begin{equation} \label{2.2}
\phi_{j}(x;z)=\Big(\sin{z\pi(1-\frac{x}{l_{j}})}+\sin{z l_{j}}\sum_{n=1}^{\infty}\frac{q_{j,n}\sin[\frac{n\pi}{l_{j}}(l_{j}-x)]}{z^{2}-(\frac{n\pi}{l_{j}})^{2}}\Big)\prod_{k\neq j}\sin{z l_{k}},\,x\in T,
\end{equation} and then we can verify
\begin{equation}
\phi_{1}(0;z)=\phi_{2}(0;z)=\cdots=\phi_{m}(0;z)=\sin z\pi;
\end{equation}
\begin{equation}
\phi_{1}(l_{1};z)=\phi_{2}(l_{2};z)=\cdots=\phi_{m}(l_{m};z)=0.
\end{equation}
Thus, one point interaction at $v_{0}=0$ in~(\ref{1.5}) is satisfied.
Moreover,
\begin{equation}\label{3333}
\phi_{j}'(x;z)=\Big(-\frac{z\pi}{l_{j}}\cos{z\pi(1-\frac{x}{l_{j}})}-\sin{z l_{j}}\sum_{n=1}^{\infty}\frac{n\pi}{l_{j}}\frac{q_{j,n}\cos[\frac{n\pi}{l_{j}}(l_{j}-x)]}{z^{2}-(\frac{n\pi}{l_{j}})^{2}}\Big)\prod_{k\neq j}\sin{z l_{k}},\,x\in T,
\end{equation} 
and then
\begin{eqnarray}\label{2.3}
\phi_{j}'(0;z)=\Big(-\frac{z\pi}{l_{j}}\cos{z\pi}-\sin{z l_{j}}\sum_{n=1}^{\infty}\frac{(-1)^{n}n\pi}{l_{j}}\frac{q_{j,n}}{z^{2}-(\frac{n\pi}{l_{j}})^{2}}\Big)\prod_{k\neq j}\sin{z l_{k}},\,j=1,2,\ldots,m.
\end{eqnarray}
To prepare a Kirchhoff condition at $x=v_{0}$, we sum over $j$-th summands.
\begin{eqnarray}
\sum_{j=1}^{m}\phi_{j}'(0;z)=-\sum_{j=1}^{m}\Big(\frac{z\pi}{l_{j}}\cos{z\pi}+\sin{z l_{j}}\sum_{n=1}^{\infty}\frac{(-1)^{n}n\pi}{l_{j}}\frac{q_{j,n}}{z^{2}-(\frac{n\pi}{l_{j}})^{2}}\Big)\prod_{k\neq j}\sin{z l_{k}}.\label{25}
\end{eqnarray}
\par
Due to the presence of $\{\bar{e}_{j}\}_{j=1}^{m}$ and $\{v_{j}\}_{j=1}^{m}$ on $\overline{T}$, we add up the special solutions on each $\bar{e}_{j}$:
\begin{equation}
\phi_{j}(x;z)=\sin{z\pi}(1-\frac{x}{\bar{l}_{j}})\prod_{k\neq j}\sin{z l_{k}},   0<x\leq\bar{l}_{j},\label{555}
\end{equation}
in which we consider the parametrization of $x$ as its are length when $x$ moves along $\cup_{j=1}^{m}\bar{e}_{j}$ that is a closed cyclic graph, and $v_{m+1}=v_{1}$. Therefore, we define $\phi_{j}(0;z):=\phi_{j-1}(\bar{l}_{j-1};z)=0$ which is zero due to~(\ref{555}).
The zero set of $\sin{z\pi}(1-\frac{x}{\bar{l}_{j}})$ is $x\in\{\frac{\bar{l}_{j}(z-n)}{z}\}_{j=1}^{m}$ and $z\in\mathbb{Z}$. Most importantly, we obtain the continuity condition at $\{v_{j}\}_{j=1}^{m}$.
Hence,
\begin{equation}
\phi_{j}(0;z)=\phi_{j+1}(0;z)=0.
\end{equation}
We also observe $$-\phi_{j}''(x;z)=(\frac{z\pi}{\bar{l}_{j}})^{2}\phi_{j}(x;z),\,x\in \bar{e}_{j}.$$ The Dirichlet problem has the general solution. To prepare the Kirchhoff condition for $x$ at vertex $v^{j}$, $j=1,2,\cdots,m,$
we use~(\ref{3333}) and~(\ref{555}) to deduce that
\begin{eqnarray}\nonumber
&&\hspace{12pt}\frac{d}{dx}\phi_{j'}(l_{j'}^{-},z)+\frac{d}{dx}\phi_{j'-1}(\bar{l}_{j'-1}^{-},z)+\frac{d}{dx}\phi_{j'}(0^{+},z)\\
&&\nonumber=-\Big(\frac{z\pi}{l_{j'}}+\frac{z\pi}{\bar{l}_{j'-1}}+\frac{z\pi}{\bar{l}_{j'}}\cos{z\pi}+\sin{zl_{j}}\sum_{n=1}^{\infty}\frac{n\pi}{l_{j}}\frac{q_{j,n}}{z^{2}-(\frac{n\pi}{l_{j}})^{2}}\Big)\prod_{k\neq j}\sin{z l_{k}},\,j=1,2,\ldots,m.\label{288}
\end{eqnarray}
Secondly, we recall the non-local condition in~(\ref{1.5}) on star graph $T$,
\begin{eqnarray}\nonumber
&&\int_{0}^{l_{j}}\phi_{j}(x;z)\overline{q_{j}(x)}dx\\&=&\Big(\int_{0}^{l_{j}}\sin{z\pi(1-\frac{x}{l_{j}})}\overline{q_{j}(x)}dx+\sin{z l_{j}}\sum_{n=1}^{\infty}\frac{q_{j,n}\int_{0}^{l_{j}}\sin[\frac{n\pi}{l_{j}}(l_{j}-x)]\overline{q_{j}(x)}dx}{z^{2}-(\frac{n\pi}{l_{j}})^{2}}\Big)\prod_{k\neq j}\sin{z l_{k}}.
\end{eqnarray}
Moreover, on the edges $\{e_{j}\}$, we consider the summation
\begin{eqnarray}\nonumber
&&\sum_{j=1}^{m}\int_{0}^{l_{j}}\phi_{j}(x;z)\overline{q_{j}(x)}dx\\&=&\sum_{j=1}^{m}\Big(\int_{0}^{l_{j}}\sin{z\pi(1-\frac{x}{l_{j}})}\overline{q_{j}(x)}dx+\sin{z l_{j}}\sum_{n=1}^{\infty}\frac{q_{j,n}\int_{0}^{l_{j}}\sin[\frac{n\pi}{l_{j}}(l_{j}-x)]\overline{q_{j}(x)}dx}{z^{2}-(\frac{n\pi}{l_{j}})^{2}}\Big)\prod_{n\neq j}\sin{z l_{k}}.\label{29}
\end{eqnarray}
On the extended edges $\{\bar{e}_{j}\}$, $\overline{q_{j}(x)}\equiv 0$. Therefore, the equation~(\ref{2.3}),~(\ref{25}),~(\ref{288}), and~(\ref{29}) add up to be the characteristic function in this paper:
\begin{eqnarray}\nonumber
\Phi(z)&\hspace{-5pt}:=\hspace{-5pt}&\sum_{j=1}^{m}\Big(\int_{0}^{l_{j}}\sin{z\pi(1-\frac{x}{l_{j}})}\overline{q_{j}(x)}dx+\sin{z l_{j}}\sum_{n=1}^{\infty}\frac{q_{j,n}\overline{q_{j,n}}}{z^{2}-(\frac{n\pi}{l_{j}})^{2}}\Big)\prod_{k\neq j}\sin{z l_{k}}\\&&
+\sum_{j=1}^{m}\Big(\frac{z\pi}{l_{j}}\cos{z\pi}+\sin{z l_{j}}\sum_{n=1}^{\infty}\frac{(-1)^{n}n\pi}{l_{j}}\frac{q_{j,n}}{z^{2}-(\frac{n\pi}{l_{j}})^{2}}\Big)\prod_{k\neq j}\sin{z l_{k}}\nonumber\\
&&+\sum_{j'=1}^{m}\Big(\frac{z\pi}{l_{j'}}+\frac{z\pi}{\bar{l}_{j'-1}}+\frac{z\pi}{\bar{l}_{j'}}\cos{z\pi}+\sin{zl_{j'}}\sum_{n=1}^{\infty}\frac{n\pi}{l_{j'}}\frac{q_{j,n}}{z^{2}-(\frac{n\pi}{l_{j'}})^{2}}\Big)\prod_{k\neq j'}\sin{z l_{k}},\label{Phi}
\end{eqnarray}
in which we set $\bar{l}_{0}=\bar{l}_{m}$.

\section{Results}
We need to complete the system~(\ref{1.5}) firstly.
\begin{lemma}
The set $\mathbb{Z}$ satisfied the first condition of~(\ref{1.5}).
\end{lemma}
\begin{proof}
It suffices to verify the first condition as~(\ref{1.1}). We proceed with~(\ref{2.2}) to obtain
\begin{equation}
-\phi''_{j}(x;z)=\Big((\frac{z\pi}{l_{j}})^{2}\sin{z\pi(1-\frac{x}{l_{j}})}+(\frac{n\pi}{l_{j}})^{2}\sin{z l_{j}}\sum_{n=1}^{\infty}\frac{q_{j,n}\sin[\frac{n\pi}{l_{j}}(l_{j}-x)]}{z^{2}-(\frac{n\pi}{l_{j}})^{2}}\Big)\prod_{k\neq j}\sin{z l_{k}},\label{331}
\end{equation}
\begin{equation}
q_{j}(x)\phi_{j}(0;z)=q_{j}(x)\Big(\sin{z\pi}+\sin{z l_{j}}\sum_{n=1}^{\infty}\frac{q_{j,n}\sin[n\pi]}{z^{2}-(\frac{n\pi}{l_{j}})^{2}}\Big)\prod_{k\neq j}\sin{z l_{k}},\label{332}
\end{equation}
and
\begin{equation}
(\frac{z\pi}{l_{j}})^{2}\phi_{j}(x;z)=(\frac{z\pi}{l_{j}})^{2}\Big(\sin{z\pi(1-\frac{x}{l_{j}})}+\sin{z l_{j}}\sum_{n=1}^{\infty}\frac{q_{j,n}\sin[\frac{n\pi}{l_{j}}(l_{j}-x)]}{z^{2}-(\frac{n\pi}{l_{j}})^{2}}\Big)\prod_{k\neq j}\sin{z l_{k}}.\label{333}
\end{equation}
Most importantly, for $z\in\mathbb{Z}$, we obtain
$$-\phi''_{j}(x;z)+q_{j}(x)\phi_{j}(0;z)=(\frac{z\pi}{l_{j}})^{2}\phi_{j}(x;z).$$
Then,~(\ref{1.1}) is satisfied on $T$. 
\end{proof}

\begin{proposition}\label{T3.2}
Let the differential system~(\ref{1.5}) be defined on the star graph 
$\overline{T}^{\sigma}$, $\sigma=1,2,$ with the identical potential function on each  edge $e_{j}$ with length $l_{j}$, $1\leq j\leq m$. Let $\Phi^{\sigma}(z)$ be the characteristic defined on $\overline{T}^{\sigma}$, $\sigma=1,2,$ If $ \Phi^{1}(z)=\Phi^{2}(z)$ for $z\in\mathbb{Z}$ for all $j$'s, then $\theta^{1}_{j}=\theta^{2}_{j}$ for $1\leq j\leq m$.
\end{proposition}
\begin{proof}
Under the assumption of proposition, we derive from~(\ref{Phi}) that
\begin{equation}
\Phi^{1}(z)-\Phi^{2}(z)=\sum_{j'=1}^{m}
\{\frac{z\pi}{\bar{l}_{j'-1}^{1}}-\frac{z\pi}{\bar{l}_{j'-1}^{2}}+\frac{z\pi}{\bar{l}_{j'}^{1}}\cos{z\pi}-\frac{z\pi}{\bar{l}_{j'}^{2}}\cos{z\pi}\}\prod_{k\neq j'}\sin{z l_{k}}=0,\,z\in\mathbb{Z}.\label{3.4}
\end{equation}
Most importantly, we set
\begin{equation}
\mathcal{Z}_{j'}:=\big\{z|\,\sin{z\pi}\prod_{k\neq j'}\sin{z l_{k}}=0\big\}.
\end{equation}
For each $j'=1,2,\cdots,m$, we plug in $\mathcal{Z}_{j'}$ consecutively into~(\ref{3.4}) to obtain
\begin{equation}
\frac{1}{\bar{l}_{j'-1}^{1}}-\frac{1}{\bar{l}_{j'-1}^{2}}+\frac{1}{\bar{l}_{j'}^{1}}(-1)^{n}-\frac{1}{\bar{l}_{j'}^{2}}(-1)^{n}=0,\,z\in\mathcal{Z}_{j'}.\label{3.6}
\end{equation}
Because $n$ can be either odd or even integers, we have shown $\bar{l}^{1}_{j'-1}=\bar{l}^{2}_{j'-1}$, $1\leq j'\leq m$.
Thus,
we prove that that all extended edges are equal from both sides. This is the uniqueness of $\{\theta_{j}\}_{j=1}^{m}$.

\end{proof}
One step further, if the knowledge of $\{l_{j}\}_{j=1}^{m}$ and $\{\bar{l}_{j}\}_{j=1}^{m}$ are given, may we recover the information on $\{q_{j}\}_{j=1}^{m}$?
\begin{lemma}\label{LL33}
The following identities hold.
\begin{equation}
\sin{z l_{j}}\sum_{n=1}^{\infty}\frac{(-1)^{n}n\pi}{l_{j}}\frac{q_{j,n}}{z^{2}-(\frac{n\pi}{l_{j}})^{2}}=\int_{0}^{\pi}q(\frac{l_{j}}{\pi}(\pi-x)\sin{zx}dx.
\end{equation}
If $q_{j,n}=0$, $n\in\mathbb{Z}$, then $q_{j}(x)\equiv0$, $0\leq j\leq m$.
\end{lemma}
\begin{proof}
We recall that 
$$q_{j}(x)=\sum_{n=1}^{\infty}q_{j,n}\sin[\frac{n\pi}{l_{j}} (l_{j}-x)]$$ with Fourier coefficients
$$q_{j,n}=\frac{2}{l_{j}}\int_{0}^{l_{j}}q_{j}(x)\sin[\frac{n\pi}{l_{j}} (l_{j}-x)]dx.$$
Using
\begin{eqnarray}
\frac{q_{j,n}}{z^{2}-(\frac{n\pi}{l_{j}})^{2}}=\frac{-\frac{1}{2}\frac{l_{j}}{n\pi}q_{j,n}}{(z+\frac{n\pi}{l_{j}})}+\frac{\frac{1}{2}\frac{l_{j}}{n\pi}q_{j,n}}{(z-\frac{n\pi}{l_{j}})},
\end{eqnarray}
we deduce
\begin{eqnarray}\nonumber
\sin{z l_{j}}\sum_{n=1}^{\infty}\frac{(-1)^{n}n\pi}{l_{j}}\frac{q_{j,n}}{z^{2}-(\frac{n\pi}{l_{j}})^{2}}&\hspace{-4pt}=\hspace{-4pt}&\sin{z l_{j}}\sum_{n=1}^{\infty}(-1)^{n}\frac{-\frac{1}{2}q_{j,n}}{(z+\frac{n\pi}{l_{j}})}+\sin{z l_{j}}\sum_{n=1}^{\infty}(-1)^{n}\frac{\frac{1}{2}q_{j,n}}{(z-\frac{n\pi}{l_{j}})}\\\nonumber
&\hspace{-4pt}=\hspace{-4pt}&\sin{z l_{j}}\sum_{n=-1}^{\infty}(-1)^{n}{l_{j}}\frac{\frac{1}{2}q_{j,n}}{(z-\frac{n\pi}{l_{j}})}+\sin{z l_{j}}\sum_{n=1}^{\infty}(-1)^{n}\frac{\frac{1}{2}q_{j,n}}{(z-\frac{n\pi}{l_{j}})}
\\&\hspace{-4pt}=\hspace{-4pt}&\sin{z l_{j}}\sum_{n=-\infty}^{\infty}(-1)^{n}\frac{\frac{1}{2}q_{j,n}}{(z-\frac{n\pi}{l_{j}})}\nonumber
\\&\hspace{-4pt}=\hspace{-4pt}&\sin{\eta\pi}\sum_{n=-\infty}^{\infty}\frac{(-1)^{n}}{\pi}\frac{\frac{1}{2\pi}q_{j,n}}{(\frac{zl_{j}}{\pi}-n)}.
\end{eqnarray}
Let $z=\frac{\eta\pi}{l_{j}}$, $\eta\in\mathbb{C}$, and $x=\frac{l_{j}}{\pi}y$. Therefore,
\begin{eqnarray}\nonumber
&&\sin{z l_{j}}\sum_{n=1}^{\infty}\frac{(-1)^{n}n\pi}{l_{j}}\frac{q_{j,n}}{z^{2}-(\frac{n\pi}{l_{j}})^{2}}\\&\hspace{-4pt}=&\sin{\eta\pi}\sum_{n=-\infty}^{\infty}\frac{(-1)^{n}}{\pi}\frac{\frac{1}{2\pi}q_{j,n}}{(\frac{zl_{j}}{\pi}-n)}\nonumber
\\&\hspace{-4pt}=&\sin{\eta\pi}\sum_{n=-\infty}^{\infty}(-1)^{n}\frac{\frac{1}{2\pi}\int_{0}^{\pi}q(\frac{l_{j}}{\pi}(\pi-x)\sin{nx}dx}{\pi(\eta-n)},\label{3100}
\end{eqnarray}
which is a Fourier sine transform as we apply the interpolation theory in complex analysis referring to M. Cartwright and Levin \cite[p.\,150,\,151]{Levin2}, and
$\{\frac{1}{2\pi}\int_{0}^{\pi}q(\frac{l_{j}}{\pi}(\pi-x)\sin{nx}dx\}_{n}\in l^{2}.$
A Fourier transform vanishes at $\mathbb{Z}$ must be trivial \cite[p.\,150,\,Theorem\,1]{Levin2}. The lemma is proven.

\end{proof}
\begin{theorem}\label{T3.3}
Let the differential system~(\ref{1.5}) be defined on the star graph 
$\overline{T}^{\sigma}$, $\sigma=1,2,$ with the potential function $q^{\sigma}_{j}$ defined on each  edge $e_{j}$ with length $l_{j}$, $1\leq j\leq m$. Let $\Phi^{\sigma}(z)$ be the characteristic defined on $\overline{T}^{\sigma}$, $\sigma=1,2.$ If $ \Phi^{1}(z)=\Phi^{2}(z)$ for $z\in\frac{\mathbb{Z}\pi}{l_{j}}$, then $q^{1}_{j}=q^{2}_{j}$ for $1\leq j\leq m$.
\end{theorem}
\begin{proof}
We consider the subtraction from~(\ref{Phi}).
\begin{eqnarray}\nonumber
\Phi^{1}(z)-\Phi^{2}(z)&\hspace{-5pt}=\hspace{-5pt}&\sum_{j=1}^{m}\Big(\int_{0}^{l_{j}}\sin{z\pi(1-\frac{x}{l_{j}})}\overline{q^{1}_{j}(x)-q^{2}_{j}(x)}dx+\sin{z l_{j}}\sum_{n=1}^{\infty}\frac{q^{1}_{j,n}\overline{q^{1}_{j,n}}-q^{2}_{j,n}\overline{q^{2}_{j,n}}}{z^{2}-(\frac{n\pi}{l_{j}})^{2}}\Big)\prod_{k\neq j}\sin{z l_{k}}\\&&
+\sum_{j=1}^{m}\Big(\sin{z l_{j}}\sum_{n=1}^{\infty}\frac{(-1)^{n}n\pi}{l_{j}}\frac{q^{1}_{j,n}-q^{2}_{j,n}}{z^{2}-(\frac{n\pi}{l_{j}})^{2}}\Big)\prod_{k\neq j}\sin{z l_{k}}\nonumber\\
&&+\sum_{j=1}^{m}\Big(\sin{zl_{j}}\sum_{n=1}^{\infty}\frac{n\pi}{l_{j}}\frac{q^{1}_{j,n}-q^{2}_{j,n}}{z^{2}-(\frac{n\pi}{l_{j}})^{2}}\Big)\prod_{k\neq j}\sin{z l_{k}}.\label{3.8}
\end{eqnarray}
Now we set
\begin{equation}
\mathbf{Z}_{j}:=\{z|\,\sin{z l_{j}}\prod_{k\neq j}\sin{z l_{k}}=0\}.
\end{equation}
For each $j$, let us plug in $\mathbf{Z}_{j}$ into~(\ref{3.8}) consecutively under the assumption of theorem  to deduce
\begin{eqnarray}\nonumber
\Phi^{1}(z)-\Phi^{2}(z)&\hspace{-5pt}=\hspace{-5pt}&\int_{0}^{l_{j}}\sin{z\pi(1-\frac{x}{l_{j}})}\overline{q^{1}_{j}(x)-q^{2}_{j}(x)}dx+\sin{z l_{j}}\sum_{n=1}^{\infty}\frac{q^{1}_{j,n}\overline{q^{1}_{j,n}}-q^{2}_{j,n}\overline{q^{2}_{j,n}}}{z^{2}-(\frac{n\pi}{l_{j}})^{2}}\\&&
+\sin{z l_{j}}\sum_{n=1}^{\infty}\frac{(-1)^{n}n\pi}{l_{j}}\frac{q^{1}_{j,n}-q^{2}_{j,n}}{z^{2}-(\frac{n\pi}{l_{j}})^{2}}+\sin{zl_{j}}\sum_{n=1}^{\infty}\frac{n\pi}{l_{j}}\frac{q^{1}_{j,n}-q^{2}_{j,n}}{z^{2}-(\frac{n\pi}{l_{j}})^{2}}=0,\,z\in\mathbf{Z}_{j}.
\end{eqnarray}
Now, we apply Lemma \ref{LL33} to deduce that
\begin{equation}
\overline{q^{1}_{j,n}}-\overline{{q}^{2}_{j,n}}+(-1)^{n}\frac{l_{j}}{n\pi}(q^{1}_{j,n}\overline{q^{1}_{j,n}}-q^{2}_{j,n}\overline{q^{2}_{j,n}})
+q^{1}_{j,n}-q^{2}_{j,n}+(-1)^{n}(q^{1}_{j,n}-q^{2}_{j,n})=0.\label{3.16}
\end{equation}
For odd $n$,~(\ref{3.16}) is deduced to be 
\begin{equation}
\overline{q^{1}_{j,n}}-\overline{{q}^{2}_{j,n}}+(-1)^{n}\frac{l_{j}}{n\pi}(q^{1}_{j,n}\overline{q^{1}_{j,n}}-q^{2}_{j,n}\overline{q^{2}_{j,n}})=0.
\end{equation}
Therefore, we obtain $\Im \{\overline{q^{1}_{j,n}}-\overline{{q}^{2}_{j,n}}\}=0$.
Using $$q_{j,n}=\frac{2}{l_{j}}\int_{0}^{l_{j}}\Re\{q_{j}(x)\}\sin[\frac{n\pi}{l_{j}} (l_{j}-x)]dx+i\frac{2}{l_{j}}\int_{0}^{l_{j}}\Im\{q_{j}(x)\}\sin[\frac{n\pi}{l_{j}} (l_{j}-x)]dx.$$
Hence, we deduce $\Im\{q_{j}^{1}(x)-q_{j}^{2}(x)\}=0$. For even $n$, we obtain that $3\Re\{\overline{q^{1}_{j,n}}-\overline{{q}^{2}_{j,n}}\}+i\Im\{\overline{q^{1}_{j,n}}-\overline{{q}^{2}_{j,n}}\}$ is real, so  $\Im\{q_{j}^{1}(x)-q_{j}^{2}(x)\}\equiv0$. 
Using~(\ref{2.1})
\begin{eqnarray}\label{3.18}
\mathcal{F}(q^{1}_{j}-q^{2}_{j})(z)=\frac{2}{l_{j}}\int_{0}^{l_{j}}(q^{1}_{j}(x)-q^{2}_{j}(x))\sin[\frac{z\pi}{l_{j}} (l_{j}-x)]dx.
\end{eqnarray}
Hence, $\mathcal{F}(q^{1}_{j}-q^{2}_{j})$ is real on the real line. Using Cauchy-Riemann equation and Mean Value Theorem, $\mathcal{F}(q^{1}_{j}-q^{2}_{j})(z)$ is a constant on the real line. The constant must zero considering~(\ref{3.18}) at $z=0$. Therefore, the theorem is proven.


\end{proof}

\begin{acknowledgement}
The author ingenuously appreciates the research funding supported by NSTC under the project number 113-2115-M-131 -001. The content of this manuscript does not necessarily reflect the position or the policy of the administration, and no official endorsement should be inferred, neither.
\end{acknowledgement}

\end{document}